\documentclass{llncs}

\usepackage{etex,color,verbatim,siunitx,bussproofs}
\usepackage{amssymb,amsmath,url,charter,latexsym}
\usepackage{tikz}
\usetikzlibrary{backgrounds,positioning,trees,shapes}
\usepackage{hyperref}
\hypersetup{pdfborder = {0 0 0}, breaklinks = true,
pdftitle={Knowing Values and Public Inspection},
pdfauthor={Jan van Eijck, Malvin Gattinger, Yanjing Wang}}

\newcommand{\PIL}{\textsf{PIL}}
\newcommand{\K}{\ensuremath{\textit{K}}}
\newcommand{\Kv}{\ensuremath{\textit{Kv}}}
\newcommand{\C}{\mathbb{C}}
\newcommand{\D}{\mathcal{D}}
\newcommand{\M}{\mathcal{M}}
\newcommand{\SPIL}{\mathbb{SPIL}}
\newcommand{\TAUT}{{\texttt{TAUT}}}
\newcommand{\DIST}{{\texttt{DIST}}}
\newcommand{\LS}{{\texttt{LEARN}}}
\newcommand{\NF}{{\texttt{NF}}}
\newcommand{\tr}[1]{\text{#1}}
\newcommand{\COM}{{\texttt{COMM}}}
\newcommand{\IR}{{\texttt{IR}}}
\newcommand{\RIR}{{\texttt{RIR}}}
\newcommand{\DET}{\ensuremath{\texttt{DET}}}
\newcommand{\MP}{{\texttt{MP}}}
\newcommand{\NEC}{{\texttt{NEC}}}
\newcommand{\Kf}{\ensuremath{\mathcal{K}\!f}}

\begin{document}

\title{Knowing Values and Public Inspection}
\titlerunning{Knowing Values and Public Inspection}
\author{Jan van Eijck\inst{1,2},
Malvin Gattinger\inst{1},
Yanjing Wang\inst{3}}
\tocauthor{van Eijck et al.}
\institute{ILLC, University of Amsterdam, Amsterdam, The Netherlands \and
SEN1, CWI, Amsterdam, The Netherlands \and
Department of Philosophy, Peking University, Beijing, China}
\maketitle

\begin{abstract}
We present a basic dynamic epistemic logic of ``knowing the value''.
Analogous to public announcement in standard DEL, we study ``public
inspection'', a new dynamic operator which updates the agents'
knowledge about the values of constants. We provide a sound and
strongly complete axiomatization for the single and multi-agent
case, making use of the well-known Armstrong axioms for dependencies
in databases.
\end{abstract}

\section{Introduction}

Standard epistemic logic studies propositional knowledge expressed by
``knowing that''. However, in everyday life we talk about knowledge
in many other ways, such as ``knowing what the password  is'',
``knowing how to swim'', ``knowing why he was late'' and so on.
Recently the epistemic logics of such expressions are drawing more
and more attention (see \cite{Wang16} for a survey).

Merely reasoning about static knowledge is important but it is also
interesting to study the changes of knowledge. Dynamic Epistemic
Logic (DEL) is an important tool for this, which handles how
knowledge (and belief) is updated by events or actions
\cite{DitHoekKooi2007:del}.
For example, extending standard epistemic logic, one can update the
propositional knowledge of agents by making propositional
announcements. They are nicely studied by public announcement logic
\cite{Plaza2007:LoPC} which includes reduction axioms to completely
describe the interplay of ``knowing that'' and ``announcing that''.
Given this, we can also ask: What are natural dynamic counterparts
the knowledge expressed by other expressions such as knowing what,
knowing how etc.? How do we formalize ``announcing what''?

In this paper, we study a basic dynamic operation which updates the
knowledge of the values of certain constants.\footnote{In this paper,
by \textit{constant} we mean something which has a single value given
the actual situation. The range of possible values of a constant may
be infinite. This terminology is motivated by first-order modal logic
as it will become more clear later.}
The action of \textit{public inspection} is the knowing value
counterpart of public announcement and we will see that it fits well
with the logic of knowing value.
As an example, we may use a sensor to measure the current temperature
of the room. It is reasonable to say that after using the sensor you
will know the temperature of the room. Note that it is not reasonable
to encode this by standard public announcement since it may result in
a possibly infinite formula:
  $[t=\SI{27.1}{\degreeCelsius}]\K(t=\SI{27.1}{\degreeCelsius})\land [t=\SI{27.2}{\degreeCelsius}]\K(t=\SI{27.2}{\degreeCelsius})\land \dots$,
and the inspection action itself may require an infinite action model
in the standard DEL framework of \cite{BalMosSol98:tlopa} with a
separate event for each possible value.
Hence public inspection can be viewed as a public announcement of the
actual value, but new techniques are required to express it formally.
In our simple framework we define knowing and inspecting values as
primitive operators, leaving the actual values out of our logical
language.

The notions of knowing and inspecting values have a natural
connection with dependencies in databases. This will also play a
crucial role in the later technical development of the paper.
In particular, our completeness proofs employ the famous set of
axioms from \cite{armstrong1974dependency}.
For now, consider the following example.

\begin{example}
Suppose a university course was evaluated using anonymous
questionnaires which besides an assessment for the teacher also
asked the students for their main subject.
See Table \ref{table:ExampleEvaluation} for the results.
Now suppose a student tells you, the teacher, that his major is
Computer Science. Then clearly you know how that student assessed
the course, since there is some dependency between the two columns.
More precisely, in the cases of students 3 and 4, telling you the
value of ``Subject'' effectively also tells you the value of
``Assessment''. In practice, a better questionnaire would only ask
for combinations of questions that do not allow the identification
of students.
\begin{table}
\centering
\vspace{-1em}
\begin{tabular}{llll}
Student & Subject & Assessment \\
\hline
1 & Mathematics & good \\
2 & Mathematics & very good \\
3 & Logic       & good \\
4 & Computer Science & bad \\
\end{tabular}
\caption{Evaluation Results}
\label{table:ExampleEvaluation}
\vspace{-4em}
\end{table}
\end{example}

Other examples abound: The author of \cite{Sweeney2015:oy} gives
an account of how easily so-called `de-identified data' produced
from medical records could be `re-identified', by matching patient
names to publicly available health data.

These examples illustrate that reasoning about knowledge of values in
isolation, i.e.~separated from knowledge \emph{that}, is both possible
and informative.  It is such knowledge and its dynamics that we will
study here.

\section{Existing Work}

Our work relates to a collection of papers on epistemic logics with
other operators than the standard ``knowing that'' $\K\varphi$. In
particular we are interested in the $\Kv$ operator expressing that
an agent knows a value of a variable or constant. This operator is
already mentioned in the seminal work \cite{Plaza2007:LoPC} which
introduced public announcement logic (PAL). However, a complete
axiomatization of PAL together with $\Kv$ was only given in
\cite{WangFan2013KvPAL,WangFan2014CondKWhat} using the relativized
operator $\Kv(\varphi,c)$ for the single and multi-agent cases.
Moreover, it has been shown in \cite{GuWang2016KvNormal} that by
treating the negation of $\Kv$ as a primitive diamond-like operator,
the logic can be seen as a normal modal logic in disguise with
binary modalities.

Inspired by a talk partly based on an earlier version of this paper,
Baltag proposed the very expressive Logic of Epistemic Dependency
(LED) \cite{Baltag2016:KVV}, where knowing that, knowing value,
announcing that, announcing value can all be encoded in a general
language which also includes equalities like $c=4$ to facilitate
the axiomatization.

In this paper we go in the other direction: Instead of extending the
standard PAL framework with $\Kv$, we study it in isolation together
with its dynamic counterpart $[c]$ for public inspection.
In general, the motto of our work here is to see how far one can get
in formalizing knowledge and inspection of values without going all
the way to or even beyond PAL.
In particular we do not include values in the syntax and we do not
have any nested epistemic modalities.

As one would expect, our simple language is accompanied by simpler
models and also the proofs are less complicated than existing methods.
Still we consider our Public Inspection Logic (PIL) more than a toy
logic. Our completeness proof includes a novel construction which we
call ``canonical dependency graph'' (Definition \ref{def:canonical-g-and-m}).
We also establish the precise connection between our axioms and the
Armstrong axioms widely used in database theory \cite{armstrong1974dependency}.

Table \ref{table:LanguageComparison} shows how PIL fits into the
family of existing languages. Note that \cite{Baltag2016:KVV} is
the most expressive language in which all operators are encoded
using $\K_i^{t_1, \dots, t_n}t$ which expresses that given the
current values of $t_1$ to $t_n$, agent $i$ knows the value of $t$.
Moreover, to obtain a complete proof system for LED one also needs
to include equality and rigid constants in the language. It is thus
an open question to find axiomatizations for a language between PIL
and LED without equality.

\begin{table}
\[ \arraycolsep=5pt
\begin{array}{l l l l l l l l l}
\text{PAL}       & p & \K \varphi &        &             &          & [!\varphi] \varphi &
  & \text{\cite{Plaza2007:LoPC}} \\
\text{PAL}+\Kv   & p & \K \varphi & \Kv(c) &             &          & [!\varphi] \varphi &
& \text{\cite{Plaza2007:LoPC}}\\
\text{PAL}+\Kv^r & p & \K \varphi & \Kv(c) & \Kv(\varphi,c) &          & [!\varphi] \varphi &
  & \text{\cite{WangFan2013KvPAL,WangFan2014CondKWhat,GuWang2016KvNormal}}\\
\text{PIL}       &   &         & \Kv(c) &             & [c] \varphi & &
  & \text{this paper}\\
\text{PIL}+\K    &   & \K \varphi & \Kv(c) &             & [c] \varphi & &
  & \text{future work}\\
\text{LED}       & p & \K \varphi & \Kv(c) & \Kv(\varphi,c) & [c] \varphi & [!\varphi] \varphi & c=c
  & \text{\cite{Baltag2016:KVV} }\\
\end{array} \]
\caption{Comparison of Languages}
\label{table:LanguageComparison}
\end{table}

All languages include the standard boolean operators $\top$, $\lnot$
and $\land$ which we do not list in Table \ref{table:LanguageComparison}.

We also discuss other related works not in this line at the end
of the paper.

\section{Single-Agent PIL}

We first consider a simple single-agent language to talk about
knowing and inspecting values. Throughout the paper we assume
a fixed set of constants $\C$.

\begin{definition}[Syntax]
Let $c$ range over $\C$.
The language $\mathcal{L}_1$ is given by:
\[ \varphi ::= \top  \mid  \lnot\varphi \mid \varphi\land\varphi \mid \Kv(c) \mid [c]\varphi \]
\end{definition}

Besides standard interpretations of the boolean connectives, the
intended meanings are as follows:
$\Kv(c)$ reads ``the agent knows the value of $c$'' and
the formula $[c]\varphi$ is meant to say ``after revealing the actual
value of $c$, $\varphi$ is the case''.
We also use the standard abbreviations
  $\varphi\lor\psi := \lnot(\lnot\varphi\land\lnot\psi)$ and
  $\varphi\to\psi := \lnot\varphi\lor\psi$.

\begin{definition}[Models and Semantics]\label{def:PIL-models-and-semantics}
A model for $\mathcal{L}_1$ is a tuple $\M = \langle S, \D, V \rangle$ where
$S$ is a non-empty set of worlds (also called states),
$\D$ is a non-empty domain and
$V$ is a valuation
$V : (S \times \C) \to \D$.
To denote $V(s,c)=V(t,c)$, i.e.~that $c$ has the same value at $s$
and $t$ according to $V$, we write $s =_c t$.  If this holds for
all $c \in C\subseteq \C$ we write $s =_C t$.
The semantics are as follows:
\[ \begin{array}{lll} \hline
\M,s\vDash \top  &  &   \textrm{always}\\
\M,s\vDash \neg\varphi &\Leftrightarrow& \M,s\nvDash \varphi \\
\M,s\vDash \varphi\land \psi &\Leftrightarrow&\M,s\vDash \varphi \textrm{ and } \M,s\vDash \psi \\
\M,s\vDash \Kv(c)&\Leftrightarrow&\text{for all } t \in S : s =_c t\\
\M,s\vDash [c]\varphi & \Leftrightarrow & \M|^s_c,s\vDash\varphi\\
\hline \end{array} \]
where $\M|^s_c$ is $\langle S', \D, V|_{S' \times \C} \rangle$ with $S'=\{t \in S \mid s =_c t \}$.
If for a set of formulas $\Gamma$ and a formula $\varphi$ we have that
whenever a model $\M$ and a state $s$ satisfy $\M,s \vDash \Gamma$
then they also satisfy $\M,s \vDash \varphi$, then we say that $\varphi$
follows semantically from $\Gamma$ and write $\Gamma \vDash \varphi$.
If this hold for $\Gamma = \varnothing$ we say that $\varphi$
is semantically valid and write $\vDash \varphi$.
\end{definition}

Note that the actual state $s$ plays an important role in the last
clause of our semantics: Public inspection of $c$ at $s$ reveals the
\emph{local actual} value of $c$ to the agent.
The model is restricted to those worlds which agree on $c$ with $s$.
This is different from PAL and other DEL variants based on action
models, where updates are usually defined on models directly and not
on pointed models.

We employ the usual abbreviation $\langle c \rangle \varphi$ as $\neg
[c]\neg \varphi$. Note however, that public inspection of $c$ can always
take place and is deterministic. Hence the determinacy axiom $\langle
c \rangle \varphi \leftrightarrow [c] \varphi$ is semantically valid and we
include it in the following system.

\begin{definition}
The proof system $\SPIL_1$ for $\PIL$ in the language
$\mathcal{L}_1$ consists of the following axiom schemata and rules.
If a formula $\varphi$ is provable from a set of premises $\Gamma$ we
write $\Gamma \vdash \varphi$. If this holds for $\Gamma = \varnothing$
we also write $\vdash \varphi$.

\begin{minipage}[t][][b]{0.6\textwidth}
\begin{center}
\begin{tabular}{lc}
\multicolumn{2}{l}{\textbf{Axiom Schemata}}\\[0.4em]
\TAUT & \tr{all instances of propositional tautologies}\\
\DIST & $[c](\varphi \rightarrow \psi) \rightarrow ([c]\varphi \rightarrow [c]\psi)$\\
\LS   & $[c]\Kv(c)$\\
\NF   & $\Kv(c) \to [d]\Kv(c)$\\
\DET  & $\langle c\rangle\varphi \leftrightarrow [c]\varphi$\\
\COM  & $[c][d]\varphi \leftrightarrow [d][c]\varphi$\\
\IR   & $\Kv(c) \rightarrow ([c]\varphi \to \varphi)$\\
\textbf{ }
\end{tabular}
\end{center}
\end{minipage}
\hspace{1em}
\begin{minipage}[t][][b]{0.2\textwidth}
\begin{center}
\begin{tabular}{lc}
\multicolumn{2}{l}{\textbf{Rules}}\\[0.4em]
\MP& $\dfrac{\varphi,\varphi\to\psi}{\psi}$\\
\\
\NEC &$\dfrac{\varphi}{[c]\varphi}$
\end{tabular}
\end{center}
\end{minipage}
\vspace{-1em}
\end{definition}

Intuitively, \LS\ captures the effect of the inspection;
\NF\ says that the agent does not forget;
\DET\ says that inspection is deterministic;
\COM\ says that inspections commute;
finally, \IR\ expresses that inspection does not bring any new
information if the value is known already.
Note that \DET\ says that $[c]$ is a function. It also implies
seriality which we list in the following Lemma.

\begin{lemma}\label{lem:provable}
The following schemes are provable in $\SPIL_1$:
\begin{itemize}
  \item $\langle c \rangle \top$ (seriality)
  \item $\Kv(c) \rightarrow (\varphi \to [c]\varphi)$ (\IR')
  \item $[c](\varphi \land \psi) \leftrightarrow [c]\varphi \land [c]\psi$ (\DIST')
  \item $[c_1]\dots[c_n](\varphi \to \psi) \to ([c_1]\dots[c_n]\varphi \to [c_1]\dots[c_n]\psi)$ (multi-\DIST)
  \item $[c_1]\dots[c_n](\varphi \land \psi) \leftrightarrow [c_1]\dots[c_n]\varphi \land [c_1]\dots[c_n]\psi$ (multi-\DIST')
  \item $[c_1]\dots[c_n](\Kv(c_1) \land \dots \Kv(c_n))$ (multi-\LS)
  \item $(\Kv(c_1) \land \dots \land \Kv(c_n)) \rightarrow [d_1]\dots[d_n] (\Kv(c_1) \land \dots \land \Kv(c_n))$ (multi-\NF)
  \item $(\Kv(c_1) \land \dots \land \Kv(c_n)) \rightarrow ([c_1]\dots[c_n]\varphi \to \varphi)$ (multi-\IR)
\end{itemize}
Moreover, the multi-\NEC\ rule is admissible:
If $\vdash \varphi$, then $\vdash [c_1]\dots[c_n]\varphi$.
\end{lemma}
\begin{proof}
For reasons of space we only prove three of the items and leave the
others as an exercise for the reader.
For \IR', we use \DET\ and \TAUT:
\begin{prooftree}
  \AxiomC{}
  \RightLabel{(\IR)}
  \UnaryInfC{$\Kv(c) \rightarrow ([c]\lnot\varphi \to \lnot\varphi)$}
  \RightLabel{(\DET)}
  \UnaryInfC{$\Kv(c) \rightarrow (\lnot[c]\varphi \to \lnot\varphi)$}
  \RightLabel{(\TAUT)}
  \UnaryInfC{$\Kv(c) \rightarrow (\varphi \to [c]\varphi)$}
\end{prooftree}
To show multi-\NEC, we use \DIST, \NEC\  and \TAUT.
For simplicity, consider the case where $C=\{c_1, c_2\}$.
\begin{prooftree}
  \AxiomC{}
  \RightLabel{(\DIST)}
  \UnaryInfC{$[c_2](\varphi \to \psi) \to ([c_2]\varphi \to [c_2]\psi)$}
  \RightLabel{(\NEC)}
  \UnaryInfC{$[c_1]([c_2](\varphi \to \psi) \to ([c_2]\varphi \to [c_2]\psi))$}
  \RightLabel{(\DIST, \TAUT)}
  \UnaryInfC{$[c_1][c_2](\varphi \to \psi) \to [c_1]([c_2]\varphi \to [c_2]\psi)$}
  \RightLabel{(\DIST, \TAUT)}
  \UnaryInfC{$[c_1][c_2](\varphi \to \psi) \to ([c_1][c_2]\varphi \to [c_1][c_2]\psi)$}
\end{prooftree}
For multi-\LS, we use \LS, \NEC, \COM, \DIST' and \TAUT:
\begin{prooftree}
  \AxiomC{}
  \RightLabel{(\LS)}
  \UnaryInfC{$[c_1]\Kv(c_1)$}
  \RightLabel{(\NEC)}
  \UnaryInfC{$[c_2][c_1]\Kv(c_1)$}
  \RightLabel{(\COM)}
  \UnaryInfC{$[c_1][c_2]\Kv(c_1)$}
          \AxiomC{}
          \RightLabel{(\LS)}
          \UnaryInfC{$[c_2]\Kv(c_2)$}
          \RightLabel{(\NEC)}
          \UnaryInfC{$[c_1][c_2]\Kv(c_2)$}
      \RightLabel{(\DIST', \TAUT)}
      \BinaryInfC{$[c_1]([c_2]\Kv(c_1) \land [c_2]\Kv(c_2))$}
      \RightLabel{(\DIST', \TAUT)}
      \UnaryInfC{$[c_1][c_2](\Kv(c_1) \land \Kv(c_2))$}
\end{prooftree}
\end{proof}

\begin{definition}\label{def:abbreviations}
We use the following abbreviations for any two finite sets of
constants $C=\{c_1,\dots,c_m\}$ and $D=\{d_1,\dots,d_n\}$.
\begin{itemize}
  \item $\Kv(C) := \Kv(c_1) \land \dots \land \Kv(c_m)$
  \item $[C]\varphi := [c_1]\dots [c_m]\varphi$
  \item $\Kv(C,D) := [C]\Kv(D)$.
\end{itemize}
\end{definition}

Note that by multi-\DIST' and \COM\, the exact enumeration of $C$
and $D$ in Definition \ref{def:abbreviations} do not matter modulo
logical equivalence.

In particular, these abbreviations allow us to shorten the ``multi''
items from Lemma \ref{lem:provable} to
  $\Kv(C,C)$,
  $\Kv(C) \to \Kv(D,C)$ and
  $\Kv(C) \to ([C]\varphi \to \varphi)$.
The abbreviation $\Kv(C,D)$ allows us to define dependencies and
it will be crucial in our completeness proof.
We have that:
\[ \begin{array}{c} \hline
\M,s\vDash \Kv(C,D) \Leftrightarrow \text{for all } t \in S :
  \text{if } s =_C t \text{ then } s =_D t \\
\hline \end{array} \]

\begin{definition}
Let $\mathcal{L}_2$ be the language given by
$\varphi ::= \top \mid \lnot\varphi \mid \varphi\land\varphi \mid \Kv(C,C)$.
\end{definition}

Note that this language is essentially a fragment of $\mathcal{L}_1$
due to the above abbreviation, where (possibly multiple) $[c]$
operators only occur in front of $\Kv$ operators (or conjunctions
thereof). Moreover, the next Lemma might count as a small surprise.

\begin{lemma}\label{lemma:equiexpressive}
$\mathcal{L}_1$ and $\mathcal{L}_2$ are equally expressive.
\end{lemma}

\begin{proof}
As $\Kv(\cdot,\cdot)$ was just defined as an abbreviation, we
already know that $\mathcal{L}_1$ is at least as expressive as
$\mathcal{L}_2$: we have $\mathcal{L}_2 \subseteq \mathcal{L}_1$.
We can also translate in the other direction by pushing all sensing
operators through negations and conjunctions.
Formally, let $t : \mathcal{L}_1 \to \mathcal{L}_2$ be defined by
\[\begin{array}{lcl}
\Kv(d) & \mapsto      & \Kv(\varnothing,\{d\})\\
\lnot \varphi            & \mapsto & \lnot t(\varphi)\\
\varphi \land \psi       & \mapsto & t(\varphi) \land t(\psi)\\
\end{array}
\hspace{2em}
\begin{array}{lcl}
[c] \lnot \varphi        & \mapsto & \lnot t([c]\varphi)\\
{[c]} (\varphi \land \psi) & \mapsto & t([c]\varphi) \land t([c]\psi)\\
{[c]} \top  &\mapsto & \top\\
{[c_1]} \dots [c_n] \Kv(d) & \mapsto & \Kv(\{c_1,\dots,c_n\},\{d\})
\end{array}\]
Note that this translation preserves and reflects truth because
determinacy and distribution are valid (determinacy allows
us to push $[c]$ through negations, distribution to push  $[c]$
through conjunctions). At this stage we have not yet established
completeness, but determinacy is also an axiom.
Hence we can note separately that $\varphi \leftrightarrow t(\varphi)$
is provable and that $t$ preserves and reflects provability and
consistency.
\end{proof}

\begin{example}
Note that the translation of $[c]\varphi$ formulas also depends on the
top connective within $\varphi$. For example we have
\[ \begin{array}{rcl}
t([c](\lnot\Kv (d) \land [e]\Kv (f)))
& = & t([c]\lnot\Kv (d)) \land t([c][e]\Kv (f))\\
& = & \lnot\Kv(\{c\},\{d\}) \land \Kv(\{c,e\},\{f\})
   \end{array} \]
\end{example}

The language $\mathcal{L}_2$ allows us to connect PIL to the maybe
most famous axioms about database theory and dependence logic from
\cite{armstrong1974dependency}.

\begin{lemma}\label{lemma:armstrong}
Armstrong's axioms are semantically valid and derivable in $\SPIL_1$:
\begin{itemize}
  \item $\Kv(C,D)$ for any $D \subseteq C$ (projectivity)
  \item $\Kv(C,D) \land \Kv(D,E) \rightarrow \Kv(C,E)$ (transitivity)
  \item $\Kv(C,D) \land \Kv(C,E) \rightarrow \Kv(C,D \cup E)$ (additivity)
\end{itemize}
\end{lemma}

\begin{proof}
The semantic validity is easy to check, hence we focus on the derivations.

For projectivity, take any two finite sets $D \subseteq C$.
If $D=C$, then we only need a derivation like the following which basically generalizes learning to finite sets.
\begin{prooftree}
  \AxiomC{}
  \RightLabel{(\LS)}
  \UnaryInfC{$[c_1]\Kv(c_1)$}
  \RightLabel{(\NEC)}
  \UnaryInfC{$[c_2][c_1]\Kv(c_1)$}
  \RightLabel{(\COM)}
  \UnaryInfC{$[c_1][c_2]\Kv(c_1)$}
                  \AxiomC{}
                  \RightLabel{(\LS)}
                  \UnaryInfC{$[c_2]\Kv(c_2)$}
                  \RightLabel{(\NEC)}
                  \UnaryInfC{$[c_1][c_2]\Kv(c_1)$}
          \RightLabel{(\DIST)}
          \BinaryInfC{$[c_1]( [c_2]\Kv(c_1) \land [c_2]\Kv(c_2) )$}
          \RightLabel{(\DIST)}
          \UnaryInfC{$[c_1][c_2] ( \Kv(c_1) \land \Kv(c_2) )$}
\end{prooftree}
If $D \subsetneq C$, then continue by applying \NEC\ for all elements of $C\setminus D$ to get $\Kv(C,D)$.

Transitivity follows from \IR\ and \NF\ as follows.
For simplicity, first we only consider the case where $C$, $D$ and $E$ are singletons.
\begin{prooftree}
  \AxiomC{}
  \RightLabel{(\NF)}
  \UnaryInfC{$\Kv(e) \to [c]\Kv(e)$}
  \RightLabel{(\NEC)}
  \UnaryInfC{$[d](\Kv(e) \to [c]\Kv(e))$}
  \RightLabel{(\DIST)}
  \UnaryInfC{$[d]\Kv(e) \to [d][c]\Kv(e)$}
  \RightLabel{(\COM)}
  \UnaryInfC{$[d]\Kv(e) \to [c][d]\Kv(e)$}
                  \AxiomC{}
                  \RightLabel{(\IR)}
                  \UnaryInfC{$\Kv(d) \to ([d] \Kv(e) \to \Kv(e))$}
                  \RightLabel{(\NEC)}
                  \UnaryInfC{$[c] (\Kv(d) \to ([d] \Kv(e) \to \Kv(e)))$}
                  \RightLabel{(\DIST)}
                  \UnaryInfC{$[c] \Kv(d) \to [c]([d] \Kv(e) \to \Kv(e))$}
                  \RightLabel{(\DIST)}
                  \UnaryInfC{$[c] \Kv(d) \to ([c][d] \Kv(e) \to [c]\Kv(e))$}
          \RightLabel{(\TAUT)}
          \BinaryInfC{$[c] \Kv(d) \to ( [d] \Kv(e) \to [c]\Kv(e) ) $}
\end{prooftree}

Now consider any three finite sets of constants $C=\{c_1,\dots,c_l\}$.
Using the abbreviations from Definition \ref{def:abbreviations} and
the ``multi'' rules given in Lemma \ref{lem:provable} it is easy to
generalize the proof. In fact, the proof is exactly the same with
capital letters.

Similarly, additivity follows immediately from multi-\DIST'.
\end{proof}

We can now use Armstrong's axioms to prove completeness of our logic.
The crucial idea is a new definition of a canonical dependency graph.

\begin{theorem}[Strong Completeness]\label{thm:PIL-strong-completeness}
For all sets of formulas $\Delta \subseteq \mathcal{L}_1$ and all
formulas $\varphi \in \mathcal{L}_1$, if $\Delta \vDash \varphi$, then
also $\Delta \vdash \varphi$.
\end{theorem}

\begin{proof}
By contraposition using a canonical model.
Suppose $\Delta \nvdash \varphi$.
Then $\Delta \cup \{\lnot \varphi\}$ is consistent and there is a maximally
consistent set $\Gamma \subseteq \mathcal{L}_1$ such that $\Gamma \supseteq \Delta \cup \{\lnot \varphi\} $.
We will now build a model $\M_\Gamma$ such that for the world $\C$
in that model we have $\M_\Gamma, \C \vDash \Gamma$ which implies
$\Delta \nvDash \varphi$.

\begin{definition}[Canonical Graph and Model]\label{def:canonical-g-and-m}
Let the graph $G_\Gamma := (\mathcal{P}(\C),\rightarrow)$ be given by $A \rightarrow B$ iff $\Kv(A,B) \in \Gamma$.
By Lemma \ref{lemma:armstrong} this graph has properties corresponding
to the Armstrong axioms: projectivity, transitivity and additivity.
We call a set of variables $s \subseteq \C$ \emph{closed} under
$G_\Gamma$ iff whenever $A \subseteq s$ and $A \to B$ in
$G_\Gamma$, then also $B \subseteq s$.
Then let the canonical model be $\M_\Gamma := (S, \D, V)$ where
  \[ S := \{ s \subseteq \C \mid s \text{ is closed under $G_\Gamma$} \},
    \D := \{0, 1\} \text{ and }
    V(s,c) = \left\{ \begin{array}{cl}
      0 & \text{if }c \in s \\
      1 & \text{otherwise}
    \end{array}
  \right. \]
\end{definition}

Note that our domain is just $\{0,1\}$.  This is possible because
we do not have to find a model where the dependencies hold
globally. Instead, $\Kv(C,d)$ only says that given the $C$-values at
the actual world, also the $d$ values are the same at the other
worlds.  The dependency does not need to hold between two non-actual
worlds.  This distinguishes our models from relationships as
discussed in \cite{armstrong1974dependency} where no actual world or
state is used, see Example \ref{ex:pointed-difference} below.

Given the definition of a canonical model we can now show:

\begin{lemma}[Truth Lemma]
$\M_\Gamma, \C \vDash \varphi$ iff $\varphi \in \Gamma$.
\end{lemma}

Before going into the proof, let us emphasize two peculiarities of
our truth lemma: First, the states in our canonical model are not
maximally consistent sets of formulas but sets of constants.
Second, we only claim the truth Lemma at one specific state, namely
$\C$ where all constants have value $0$. As our language does not
include nested epistemic modalities, we actually never evaluate
formulas at other states of our canonical model.

\begin{proof}[Truth Lemma]
Note that it suffices to show this for all $\varphi$ in $\mathcal{L}_2$:
Given some $\varphi \in \mathcal{L}_1$, by Lemma \ref{lemma:equiexpressive} we have that $\M_\Gamma, \C \vDash \varphi \iff \M_\Gamma, \C \vDash t(\varphi)$ because the translation preserves and reflects truth.
Moreover, we have $\varphi \in \Gamma \iff t(\varphi) \in \Gamma$, because $\varphi \leftrightarrow t(\varphi)$ is provable in $\SPIL_1$.
Hence it suffices to show that $\M_\Gamma, \C \vDash t(\varphi)$ iff $t(\varphi) \in \Gamma$, i.e.~to show the Truth Lemma for $\mathcal{L}_2$.
Again, negation and conjunction are standard, the crucial case are dependencies.

  Suppose $\Kv(C,D) \in \Gamma$.
  By definition $C\to D$ in $G_\Gamma$.
  To show $\M_\Gamma, \C \vDash \Kv(C,D)$,
    take any $t$ such that $\C =_C t$ in $\M_\Gamma$.
    Then by definition of $V$ we have $C \subseteq t$.
    As $t$ is closed under $G_\Gamma$, this implies $D \subseteq t$.
    Now by definition of $V$ we have $\C =_D t$.

  For the converse, suppose $\Kv(C,D) \not\in \Gamma$.
  Then by definition $C \not\to D$ in $G_\Gamma$.
  Now, let $t := \{ c' \in \C \mid C \to \{c'\} \text{ in } G_\Gamma \}$.
  This gives us $C \subseteq t$.
  But we also have $D \not\subseteq t$ because otherwise additivity would imply $C \to D$ in $G_\Gamma$.
  Moreover, because $G_\Gamma$ is transitive it is enough to ``go one step'' in $G_\Gamma$ to get a set that is closed under $G_\Gamma$.
  This means that $t$ is closed under $G_\Gamma$ and therefore a state in our model, i.e.~we have $t \in S$.
  Now by definition of $V$ and projectivity,  we have $\C =_C t$ but $\C \neq_D t$.
  Thus $t$ is a witness for $\M_\Gamma, \C \nvDash \Kv(C,D)$.
\end{proof}

This also finishes the completeness proof. Note that we used all
three properties corresponding to the Armstrong axioms.
\end{proof}

\begin{example}\label{ex:canonical-graph}
To illustrate the idea of the canonical dependency graph, let us
study a concrete example of what the graph and model look like.
Consider the maximally consistent set
$\Gamma = \{ \lnot \Kv(c), \lnot \Kv(d), \Kv(e), \Kv(c,d), \dots \}$.
The interesting part of the canonical graph $G_\Gamma$ then looks as
follows, where the nodes are subsets of $\{c,d,e\}$. For clarity we
only draw $\to \cap \not\subseteq$, i.e.~we omit edges given by
inclusions. For example all nodes will also have an edge going to
the $\varnothing$ node.

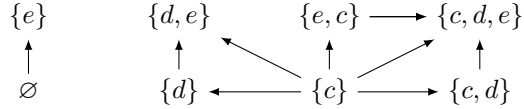
\begin{center}
\begin{tikzpicture}[node distance=2cm,>=latex,->]
  \node (cde) {$\{c,d,e\}$};
  \node (ec) [left of=cde] {$\{e,c\}$};
  \node (cd) [below of=cde, node distance=1cm] {$\{c,d\}$};
  \node (de) [left of=ec] {$\{d,e\}$};
  \node (c) [left of=cd] {$\{c\}$};
  \node (d) [left of=c] {$\{d\}$};
  \node (0) [left of=d] {$\varnothing$};
  \node (e) [left of=de] {$\{e\}$};
  \draw (c) -- (d);
  \draw (c) -- (cd);
  \draw (0) -- (e);
  \draw (d) -- (de);
  \draw (c) -- (ec);
  \draw (cd) -- (cde);
  \draw (ec) -- (cde);
  \draw (ec) -- (cde);
  \draw (c) -- (de);
  \draw (c) -- (cde);
\end{tikzpicture}
\end{center}

To get a model out of this graph, note that there are exactly three subsets of $\C$ closed under following the edges.
Namely, let $S = \{ s:\{e\}, t:\{d,e\}, u:\{c,d,e\} \}$ and use the
binary valuation which says that a constant has value 0 iff it is
an element of the state.
It is then easy to check that $\M, u \vDash \Gamma$.

\begin{center}
\begin{tabular}{cccc}
  & s & t & u \\
\hline
c & 1 & 1 & 0 \\
d & 1 & 0 & 0 \\
e & 0 & 0 & 0 \\
\end{tabular}
\end{center}
\end{example}

It is also straightforward to define an appropriate notion of bisimulation.

\begin{definition}\label{def:single-bisim}
Two pointed models $((S, \D, V),s)$ and $((S', \D', V'),s')$,
are \emph{bisimilar} iff
  (i) For all finite $C \subseteq \mathbb{C}$ and all $d \in \mathbb{C}$:
    If there is a $t \in S$ such that $s =_C t$ and $s \neq_d t$,
    then there is a $t' \in S'$ such that $s' =_C t'$ and $s' \neq_d t'$; and
  (ii) Vice versa.
\end{definition}

Note that we do not need the bisimulation to also link non-actual
worlds. This is because all formulas are evaluated at the same world.
In fact it would be too strong for the following characterization.

\begin{theorem}\label{thm:PIL-bisim}
Two pointed models satisfy the same formulas iff they are bisimilar.
\end{theorem}

\begin{proof}
By Lemma \ref{lemma:equiexpressive} we only have to consider formulas
of $\mathcal{L}_2$. Moreover, it suffices to consider formulas
$\Kv(C,d)$ with a singleton in the second set because $\Kv(C,D)$ is
equivalent to $\bigwedge_{d\in D} \Kv(C,d)$.
Then it is straightforward to show that if $\M,s$ and $\M',s'$ are
bisimilar then $\M,s\vDash\neg \Kv(C,d)\iff \M',s'\vDash\neg \Kv(C,d)$
by definition of our bisimulation. The other way around is also
obvious since the two conditions for bisimulation are based on the
semantics of  $\neg \Kv(C,d)$.
\end{proof}

Note that a bisimulation characterization for a language without the
dynamic operator can be obtained by restricting Definition
\ref{def:single-bisim} to $C=\varnothing$.
We leave it as an exercise for the reader to use this and Theorem
\ref{thm:PIL-bisim} to show that $[c]$ is not reducible, which
distinguishes it from the public announcement $[\varphi]$ in PAL.

\begin{example}[Pointed Models Make a Difference]\label{ex:pointed-difference}
It seems that the following theorem of our logic does not translate to
Armstrong's system from \cite{armstrong1974dependency}.
\[ [c](\Kv(d) \lor \Kv(e)) \leftrightarrow ([c]\Kv(d) \lor [c]\Kv(e)) \]
First, to see that this is provable, note that it follows from
determinacy and seriality.
Second, it is valid because we consider pointed models which convey
more information than a simple list of possible values.
Consider the following table which represents $4$ possible worlds.
\begin{center}
  \begin{tabular}{lll}
    $c$ & $d$ & $e$ \\
    \hline
    1 & 1 & 3 \\
    1 & 1 & 2 \\
    2 & 2 & 1 \\
    2 & 3 & 1 \\
  \end{tabular}
\end{center}
Here we would say that ``After learning $c$ we know $d$ or we know $e$.'',
i.e.~the antecedent of above formula holds. However, the consequent
only holds if we evaluate formulas while pointing at a specific
world/row: It is globally true that given $c$ we will learn $d$ or
that given $c$ we will learn $e$. But none of the two disjuncts holds
globally which would be needed for a dependency in Armstrong's sense.
Note that this is more a matter of expressiveness than of logical strength.
In Armstrong's system there is just no way to express $[c](\Kv(d) \lor \Kv(e))$.
\end{example}

\section{Multi-Agent PIL}

We now generalize the Public Inspection Logic to multiple agents.
In the language we use $\Kv_i$ to say that agent $i$ knows the value
of $c$ and in the models an accessibility relation for each
agent is added to describe their knowledge.
To obtain a complete proof system we can leave most axioms as above
but have to restrict the irrelevance axiom. Again the completeness
+proof uses a canonical model construction and a truth lemma for a
+restricted but equally expressive syntax.
The only change is that we now define a dependency graph for each
agent in order to define accessibility relations instead of
restricted sets of worlds.

\begin{definition}[Multi-Agent PIL]\label{def:multi-agent-PIL}
We fix a non-empty set of agents $I$. The language $\mathcal{L}^I_1$
of multi-agent Public Inspection Logic is given by
$$\varphi ::= \top  \mid  \lnot\varphi \mid \varphi\land\varphi \mid \Kv_i c \mid [c]\varphi$$
where $i\in I$.
We interpret it on models
$\langle S, \D, V, R \rangle$ where $S$, $\D$ and $V$ are as before
and $R$ assigns to each agent $i$ an equivalence relation $\sim_i$
over $S$.
The semantics are standard for the booleans and as follows:
\[ \begin{array}{lll} \hline
\M,s \vDash \Kv_i c  & \iff  & \forall t \in S : s \sim_i t \Rightarrow s =_c t \\
\M,s\vDash [c]\varphi & \iff & \M|^s_c,s\vDash\varphi\\
\hline \end{array} \]
where $\M|^s_c$ is
  $\langle S', \D, V|_{S' \times \C}, R|_{S' \times S'} \rangle$
with $S'=\{t \in S \mid s =_c t \}$.

Analogous to Definition \ref{def:abbreviations} we define the
following abbreviation to express dependencies known by agent $i$
and note its semantics:
\[ \Kv_i(C,D) := [c_1]\dots[c_n](\Kv_i(d_1) \land \dots \land \Kv_i(d_m)) \]
\[ \begin{array}{c} \hline
\M, s \vDash \Kv_i(C,D) \Leftrightarrow \text{for all } t \in S :
  \text{if } s \sim_i t \text{ and } s =_C t \text{ then } s =_D t \\
\hline \end{array} \]

The proof system $\SPIL$ for $\PIL$ in the language $\mathcal{L}^I_1$
is obtained by replacing each $\Kv$ in the axioms of $\SPIL_1$ by
$\Kv_i$, and replacing \IR\ by the following restricted version:
\begin{center}
\begin{tabular}{ll}
\RIR \hspace{1.5em} & $\Kv_i c \rightarrow ([c]\varphi \to \varphi)$ where $\varphi$ does not mention any agent besides $i$ \\
\end{tabular}
\end{center}
\end{definition}

Before summarizing the completeness proof for the multi-agent
setting, let us highlight some details of this definition.

As be fore the actual state $s$ plays an important role in the
semantics of $[c]$. However, we could also use an alternative but
equivalent definition: Instead of deleting states, only delete the
$\sim_i$ links between states that disagree on the value of $c$.
Then the update no longer depends on the actual state.

For traditional reasons we define $\sim_i$ to be an equivalence
relation. This is not strictly necessary, because our language
can not tell whether the relation is reflexive, transitive or
symmetric. Removing this constraint and extending the class of
models would thus not make any difference in terms of validities.

For the proof system, note that the original irrelevance axiom \IR\
is \emph{not} valid in the multi-agent setting because $\varphi$ might
talk about other agents for which the inspection of $c$ does matter.

\begin{theorem}[Strong Completeness for $\SPIL$]
For all sets of formulas $\Delta \subseteq \mathcal{L}^I_1$ and all
formulas $\varphi \in \mathcal{L}^I_1$, if $\Delta \vDash \varphi$, then
also $\Delta \vdash \varphi$.
\end{theorem}
\begin{proof}
By the same methods as for Theorem \ref{thm:PIL-strong-completeness}.
Given a maximally consistent set $\Gamma \subseteq \mathcal{L}^I_1$
we want to build a model $\M_\Gamma$ such that for the world $\C$
in that model we have $\M_\Gamma, \C \vDash \Gamma$.

First, for each agent $i \in I$, let $G_\Gamma^i$ be the graph given
by $A \rightarrow_i B \ \ :\iff  \ \Gamma \vdash \Kv_i(A,B)$.
Given that the proof system $\SPIL$ was obtained by indexing the
axioms of $\SPIL_1$, it is easy to check that indexed versions of
the Armstrong axioms are provable and therefore all the graphs
$G_\Gamma^i$ for $i \in I$ will have the corresponding properties.
In particular $\RIR$ suffices for this.

Second, define the canonical model $\M_\Gamma := (S, \D, V, R)$ where
$S := \mathcal{P}(\C)$, $\D := \{0,1\}$,
$V(s,c):=0$ if $c \in s$ and $V(s,c):=1$ otherwise,
and $s \sim_i t$ iff $s$ and $t$ are both closed or both not closed under $G_\Gamma^i$.

\begin{lemma}[Multi-Agent Truth Lemma]
$\M_\Gamma, \C \vDash \varphi$ iff $\varphi \in \Gamma$.
\end{lemma}

\begin{proof}
Again it suffices to show the Truth Lemma for a restricted language
and we only consider the state $\C$.
We proceed by induction on $\varphi$.
The crucial case is when $\varphi$ is of form $\Kv_i(C,D)$.

  Suppose $\Kv_i(C,D) \in \Gamma$.
  Then by definition $C \to D$ in $G_\Gamma^i$.
  To show $\M_\Gamma, \C \vDash \Kv_i(C,D)$,
    take any $t$ such that $\C \sim_i t$ and $\C =_C t$ in $\M_\Gamma$.
    Then by definition of $V$ we have $C \subseteq t$.
    Moreover, $\C$ is closed under $G_\Gamma^i$.
    Hence by definition of $\sim_i$ also $t$ must be closed under $G_\Gamma^i$ which implies $D \subseteq t$.
    Now by definition of $V$ we have $\C =_D t$.

  For the converse, suppose $\Kv_i(C,D) \not\in \Gamma$.
  Then by definition $C \not\to D$ in $G_\Gamma^i$.
  Now, let $t := \{ c' \in \C \mid C \to \{c'\} \text{ in } G_\Gamma^i \}$.
  This gives us $C \subseteq t$.
  But we also have $D \not\subseteq t$ because otherwise additivity would imply $C \to D$ in $G_\Gamma^i$.
  Moreover, because $G_\Gamma^i$ is transitive it is enough to ``go one step'' in $G_\Gamma^i$ to get a set that is closed under $G_\Gamma^i$.
  This means that $t$ is closed under $G_\Gamma^i$ and therefore by definition of $\sim_i$ we have $\C \sim_i t$.
  Now by definition of $V$ and projectivity, we have $\C =_C t$ but $\C \neq_D t$.
  Thus $t$ is a witness for $\M_\Gamma, \C \nvDash \Kv_i(C,D)$.
\end{proof}

Again the Truth Lemma also finishes the completeness proof.
\end{proof}

\begin{figure}\centering
\vspace{-1em}
\begin{tikzpicture}[node distance=9mm,>=latex,->]
  \node (cd) {$\{c,d\}$};
  \node (c) [left  of=cd, below of=cd] {$\{c\}$};
  \node (d) [right of=cd, below of=cd] {$\{d\}$};
  \node (0) [right of=c,  below of=c] {$\varnothing$};
  \draw (c) -- (cd);
  \draw (c) -- (d);
  \node (label) [below of=0] {$G_\Gamma^1$ (omitting $\subseteq$)};
\end{tikzpicture}
\hspace{1em}
\begin{tikzpicture}[node distance=9mm,>=latex,->]
  \node (cd) {$\{c,d\}$};
  \node (c) [left  of=cd, below of=cd] {$\{c\}$};
  \node (d) [right of=cd, below of=cd] {$\{d\}$};
  \node (0) [right of=c,  below of=c] {$\varnothing$};
  \draw (d) -- (cd);
  \draw (d) -- (c);
  \node (label) [below of=0] {$G_\Gamma^2$ (omitting $\subseteq$)};
\end{tikzpicture}
\hspace{2em}
\begin{tikzpicture}[node distance=1.5cm,>=latex]
  \node (cd) {$\begin{array}{c}\{c,d\}\\c=0\\d=0\end{array}$};
  \node (c) [left  of=cd, below of=cd] {$\begin{array}{c}\{c\}\\c=0\\d=1\end{array}$};
  \node (d) [right of=cd, below of=cd] {$\begin{array}{c}\{d\}\\c=1\\d=0\end{array}$};
  \node (0) [right of=c,  below of=c] {$\begin{array}{c}\varnothing\\c=1\\d=1\end{array}$};
  \draw (cd) edge [out=260,in=100] node[left]{2} (0);
  \draw (cd) edge [out=280,in=80] node[right]{1} (0);
  \draw (c) -- node[above]{2} (cd);
  \draw (c) -- node[above]{2} (0);
  \draw (d) -- node[above]{1} (cd);
  \draw (d) -- node[above]{1} (0);
\end{tikzpicture}
\caption{Two canonical dependency graphs and the resulting canonical model.}
\label{fig:multi-canonical-constructs}
\vspace{-2em}
\end{figure}
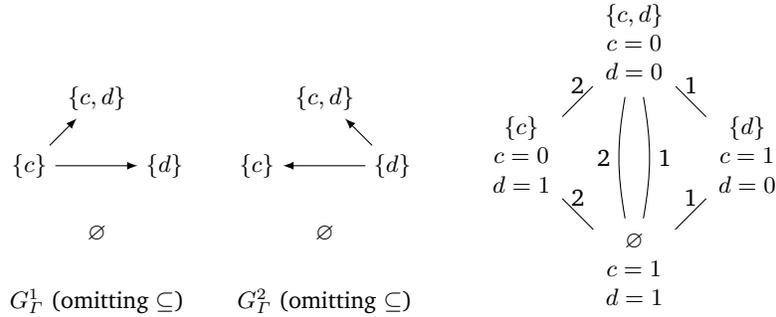

\begin{example}
Analogous to Example \ref{ex:canonical-graph}, the following
illustrates the multi-agent version of our canonical construction.
Consider the maximally consistent set
$\Gamma = \{
  \lnot \Kv_1(d), \Kv_1(c,d), \lnot \Kv_1(d,c),
  \lnot \Kv_2(c), \lnot \Kv_2(c,d), \Kv_1(d,c),
\dots \}$.
Note that agents $1$ and $2$ do not differ in which values they know
right now but there is a difference in what they will learn from
inspections of $c$ and $d$.
The two canonical dependency graphs generated from $\Gamma$ are shown
in Figure \ref{fig:multi-canonical-constructs}. Again for clarity we
only draw the non-inclusion arrows.
The subsets of $\C=\{c,d\}$ closed under the graphs are thus
  $\{ \{c,d\}, \{d\}, \varnothing \}$ and
  $\{ \{c,d\}, \{c\}, \varnothing \}$
for agent $1$ and $2$ respectively, inducing the equivalence
relations as shown in Figure \ref{fig:multi-canonical-constructs}.
\end{example}

It is also not hard to find the right notion of bisimulation for
$\SPIL$.

\begin{definition}
Given two models $(S,\D,V,R)$ and $(S',\D',V',R')$, a relation
$Z \subseteq S \times S'$ is a \emph{multi-agent bisimulation}
iff for all $s Z s'$ we have
  (i) For all finite $C \subseteq \mathbb{C}$, all $d \in \mathbb{C}$ and all agents $i$:
    If there is a $t \in S$ such that $s \sim_i t$ and $s =_C t$ and $s \neq_d t$,
    then there is a $t' \in S'$ such that $t Z t'$ and $s' \sim_i t'$ and $s =_C t$ and $s' \neq_d t'$; and
  (ii) Vice versa.
\end{definition}

\begin{theorem}
Two pointed models satisfy the same formulas of the multi-agent
language $\mathcal{L}^I_1$ iff there is a multi-agent
bisimulation linking them.
\end{theorem}

As it is very similar to the one of Theorem \ref{thm:PIL-bisim},
we omit the proof here.

\section{Future Work}

Between our specific approach and the general language of
\cite{Baltag2016:KVV}, a lot can still be explored. An advantage
of having a weaker language with explicit operators, instead of
encoding them in a more general language, is that we can clearly see
the properties of those operators showing up as intuitive axioms.

The framework can be extended in different directions. We could for
example add equalities $c=d$ to the language, together with knowledge
$\K(c=d)$ and announcement $[c=d]$. No changes to the models are
needed, but axiomatizing these operators seems not straightforward.
Alternatively, just like Plaza added $\Kv$ to PAL, we can also add
$\K$ to PIL. Another next language to be studied is thus $\PIL + \K$
from Table \ref{table:LanguageComparison} above and given by
$$\varphi ::= \top \mid \lnot\varphi \mid \varphi\land\varphi \mid \Kv_i c \mid \K_i \varphi \mid [c]\varphi.$$
Note that in this language, we can also express \emph{knowledge of}
dependency in contrast to \textit{de facto} dependency.
For example, $\K_i [c]\Kv_i d$ expresses that agent $i$ knows that
$d$ functionally depends on $c$, while $[c]\Kv_i d$ express that the
value of $d$ (given the information state of $i$) is determined by
the \textit{actual value} of $c$ \textit{de facto}. In particular the
latter does not imply that $i$ knows this. The agent can still
consider other values of $c$ possible that would not determine the
value of $d$.
To see the difference technically, we can spell out the truth
condition for $\K_i[c]\Kv_i(d)$ under standard Kripke semantics
for $\K_i$ on S5 models:
\[ \M,s\vDash \K_i[c]\Kv_i(d) \Leftrightarrow \text{ for all }t_1\sim_is, t_2\sim_i s: t_1=_{c} t_2 \implies  t_1=_{d} t_2\\ \]
Now consider Example~\ref{ex:pointed-difference}: $[c]\Kv(d)$ holds
in the first row, but $\K[c]\Kv(d)$ does not hold since the semantics
of $\K$ require $[c]\Kv(d)$ to hold at \textit{all} worlds considered
possible by the agent. This also shows that $[c]\Kv(d)$ is not
positively introspective (i.e.~the formula $ [c]\Kv(d)\to \K_i [c]\Kv(d)$
is not valid), and it is essentially not a subjective epistemic formula.

In this way, $\K [c]\Kv(d)$ can also be viewed as the atomic formula
$=\!\!(c, d)$ in \textit{dependence logic} (DL) from \cite{Depbook}.
A \textit{team model} of DL can be viewed as the set of epistemically
accessible worlds, i.e., a single-agent model in our case.
The connection with dependence logic also brings PIL closer to
the first-order variant of \textit{epistemic inquisitive logic} by
\cite{CiardelliR15}, where knowledge of entailment of interrogatives
can also be viewed as the knowledge of dependency. For a detailed
comparison with our approach, see \cite[Sec. 6.7.4]{Ivano16}.

Another approach is to make the dependency more explicit and include
functions in the syntax. In \cite{Ding15thesis} a functional
dependency operator $\Kf_i$ is added to the epistemic language with
$\Kv_i$ operators: $\Kf_i(c, d):= \exists f \K_i(d=f(c))$ where $f$
ranges over a pool of functions.

Finally, there is an independent but related line of work on
(in)dependency of variables using predicates, see for example
\cite{MoreN10,Naumov12,NaumovN14,HarNau13:FunDepStratG}.
In particular, \cite{NaumovN14} also uses a notion of dependency as
an epistemic implication ``Knowing c implies knowing d.'',
similar to our formula $\Kv(c,d)$.
In \cite{HarNau13:FunDepStratG} also a ``dependency graph'' is used
to describe how different variables, in this case payoff functions
in strategic games, may depend on each other.
Note however, that these graphs are not the same as our canonical
dependency graphs from Definition \ref{def:canonical-g-and-m}.
Our graphs are directed and describe determination between sets of
variables. In contrast, the graphs in \cite{HarNau13:FunDepStratG}
are undirected and consist of singleton nodes for each player in a
game.
We leave a more detailed comparison for a future occasion.

\subsubsection{Acknowledgements.}
We thank the following people for useful comments on this work:
Alexandru Baltag, Peter van Emde Boas, Hans van Ditmarsch,
Jie Fan, Kai Li and our anonymous reviewers.

\bibliographystyle{splncs}
\bibliography{PIL}

\end{document}